\documentclass[10pt,english]{amsart}

\usepackage[T1]{fontenc}
\usepackage[latin9]{inputenc}
\usepackage{amsmath}
\usepackage{amsthm}
\usepackage{txfonts}
\usepackage[letterpaper,margin=3cm]{geometry}
\usepackage{babel}

\theoremstyle{plain}

\newtheorem{theorem}{Theorem}[section]
\newtheorem{corollary}[theorem]{Corollary}
\newtheorem{lemma}[theorem]{Lemma}
\newtheorem{proposition}[theorem]{Proposition}
\newtheorem{definition}[theorem]{Definition}

\newcommand{\RealVect}[1]{{\mathbb R}^{#1}}

\newcommand{\cws}{\stackrel{*}{\to}} 
\newcommand{\Ud}[1]{\tilde{U}_d^{#1}} 
\newcommand{\Rk}[3]{\frac{1}{|{#1} - {#2}|^{#3}}} 
\newcommand{\Hdrest}[1]{\Hd_{#1}} 

\def\Rd{\RealVect{d}} 
\def\Rp{\RealVect{p}} 
\def\Rpp{\RealVect{p'}} 
\def\MA{{\mathcal M}(A)} 
\def\MAp{{\mathcal M}_1(A)}
\def\Es{{\mathcal E}_s} 
\def\Hd{{\mathcal H}^d} 
\def\Ld{{\mathcal L}^d} 
\def\Hdr{\Hdrest{A}} 
\def\Id{\tilde{I}_d} 
\def\Is{I_s} 

\DeclareMathOperator{\diam}{diam}
\DeclareMathOperator{\dist}{dist}
\DeclareMathOperator{\supp}{supp}

\begin{document}

\title{Riesz $s$-equilibrium measures on $d$-rectifiable sets as $s$ approaches $d$}

\author{Matthew T. Calef}
\author{Douglas P. Hardin}\thanks{This research was supported, in part,
by the U. S. National Science Foundation under grants DMS-0505756 and DMS-0808093.}
\address{M. T. Calef, D. P. Hardin:
Department of Mathematics, 
Vanderbilt University, 
Nashville, TN 37240, 
USA }
\email{Matthew.T.Calef@Vanderbilt.Edu}
\email{Doug.Hardin@Vanderbilt.Edu}

\maketitle

\begin{abstract}
Let $A$ be a compact  set in $\Rp$ of Hausdorff dimension $d$.  For $s\in(0,d)$, the Riesz $s$-equilibrium measure $\mu^s$ is the unique Borel probability measure   with support in $A$ that
minimizes  $$ \Is(\mu):=\iint\Rk{x}{y}{s}d\mu(y)d\mu(x)$$  over all such probability measures.   If  $A$ is strongly $(\Hd, d)$-rectifiable, then $\mu^s$ converges in the weak-star topology  to normalized $d$-dimensional Hausdorff measure restricted to $A$  as $s$ approaches $d$ from below.  
\end{abstract}


\section{Introduction}

Let $A$ be a compact subset of $\Rp$ with positive and finite
$d$-dimensional Hausdorff measure $\Hd(A)$. Let $\MA$ denote the set of Radon
measures with support in $A$, and $\MAp\subset\MA$ the Borel probability
measures with support in $A$. The Riesz $s$-energy of a measure
$\mu\in\MA$ is defined by
$$I_s(\mu) :=\iint \Rk{x}{y}{s} d\mu(y)d\mu(x).$$ 
If $s\in(0,d)$, then
there is a unique measure $\mu^s=\mu^{s,A}\in\MAp$  called the {\em
($s$-)equilibrium measure on $A$} such that $I_s(\mu^s)<I_s(\nu)$ for any measure
$\nu\in\MAp\backslash\{\mu^s\}$, while, for $s\geq d$,
$I_s(\nu)=\infty$ for any non-trivial measure $\nu\in\MA$ (cf.~\cite{Landkof1,Mattila1}).  The uniqueness of the equilibrium measure arises from the positivity of the Riesz kernel (cf. \cite{Gotz1, Landkof1}.  
For example, in the case that $A$ is the interval $[-1,1]$ and $s\in(0,1)$,  it is well-known (cf.~\cite{HardinSaff2}) that $d\mu^s(x)= {c_s} {(1-x^2)^{\frac{s-1}{2}}} dx$ where $c_s$ is chosen so that $\mu^s$ is a probability measure. 

In this paper, we investigate the behavior of $\mu^s$ as $s$ approaches $d$ from below.  For $A=[-1,1]$, we see directly from the above expression  that $\mu^s$ converges in the weak-star sense as $s\uparrow 1$ to normalized Lebesgue measure restricted to $A$.  It is natural to ask how general is this phenomena.  
  We are further motivated by recent results 
concerning the following related discrete minimal energy problem. For a configuration of $N\geq2$ points $\omega_N:=\{x_1,
\dots,x_N\}$ and $s>0$, the Riesz $s$-energy of $\omega_N$ is defined by 
$$
E_s(\omega_N) := \sum_{\stackrel{i,j=1}{i\ne j}}^N \Rk{x_i}{x_j}{s}.
$$
The compactness of $A$ and the lower semicontinuity of the Riesz kernel imply that there is a (not necessarily unique) 
configuration $\omega_N^s\subset A$ that minimizes $E_s$ over all $N$-point configurations on $A$. When $s<d$ the above 
continuous and discrete problems are related by the following two results (cf.~\cite{Landkof1}). First, $E_s(\omega_N^s)/N^2 \to 
I_s(\mu^s)$ as $N\to\infty$. Second, the sequence of configurations $\{\omega_N^s\}_{N=1}^\infty$ has asymptotic distribution 
$\mu^s$, that is, the sequence of discrete  measures  $$\mu^{s,N} := \frac{1}{N}\sum_{x\in\omega_N^s} 
\delta_{x}$$  (where $\delta_x$ denotes the unit atomic measure at $x$) converges to $\mu^s$ in the weak-star topology on $\MA$ as $N\to\infty$. We use a starred arrow to denote weak-star convergence, that is, for $s\in(0,d)$ we have
\begin{equation}\label{eq:motive1}
\mu^{s,N}\cws\mu^s \qquad \text{as $N\to\infty$.}
\end{equation}

In the case $s\geq d$, the discrete minimal energy problem is well-posed even though the continuous problem is not. Recently, 
asymptotic results for the discrete minimal energy problem were obtained in~\cite{HardinSaff1} and~\cite{BHS1} for 
  this range of $s$ and the case 
that $A$ is a $d$-rectifiable set where   $A$ is {\em $d$-rectifiable} (cf. \cite[\S 3.2.14]{Federer1}) if it is the Lipschitz image 
of a bounded set in $\Rd$. In this case, 
\begin{equation}\label{eq:motive2}
\mu^{s,N}\cws\Hdr/\Hd(A)\qquad \text{as $N\to \infty$.}
\end{equation} (Here and in the rest of the paper we use the notation $
\mu_E$ to denote the restriction of a measure $\mu$ to a $\mu$-measurable set $E$. e.g. $\Hdr = \Hd(\cdot \cap A)$.) For 
technical reasons, the results in~\cite{HardinSaff1} and~\cite{BHS1} for the  case $s=d$ further require that $A$ be a subset of a $d$-dimensional $C^1$ manifold, although it is conjectured that this hypothesis is unnecessary.

The limits \eqref{eq:motive1} and \eqref{eq:motive2} suggest that  
$\mu^s\cws\Hdr/\Hd(A)$ as $s\uparrow d$ whenever $A$ is  $d$-rectifiable.   If $A$ is strongly $(\Hd, d)$-rectifiable   (see Definition~\ref{df:strongrect} below), we show that this is indeed the case.  A primary tool in our work is the following normalized $d$-energy 
of a measure $$
\Id(\mu) := \lim_{s\uparrow d} (d-s)I_s(\mu),
$$
which we show is well-defined  for every measure  $\mu\in\MA$  and is uniquely minimized over $\MAp$ by the measure $\lambda^d := \Hdr/\Hd(A)$.

A map $f:A\to\Rpp$ is {\em Lipschitz} if there is a constant $L$ such that,
for any $x$, $y\in A$, $$|f(x)-f(y)|<L|x-y|,$$ 
and is {\em bi-Lipschitz} if there is a constant $L$ such that
for any $x$, $y\in A$, $$\frac{1}{L}|x-y| < |f(x)-f(y)|<L|x-y|.$$ 

A set $A\subset\Rp$ is {\em $(\Hd,d)$-rectifiable} (cf. \cite[\S 3.2.14]{Federer1})  if $\Hd(A)<\infty$ and there exists a countable collection $E_1,E_2,\dots$ of $d$-rectifiable sets that cover $\Hd$-almost all of $A$. That is, there exists a countable collection of bounded subsets of $\Rd$ $K_1,K_2,\dots$ and a corresponding collection of Lipschitz maps, $\varphi_1:K_1\to\Rp,\varphi_2:K_2\to\Rp,\dots$ such that $$\Hd\left(A\backslash \bigcup_{i=1}^\infty \varphi_i(K_i)\right) = 0.$$    Moreover, it is a result of Federer \cite[\S 3.2.18]{Federer1}) that if $A$ is $(\Hd, d)$-rectifiable then for every $\varepsilon>0$, the Lipschitz maps and the bounded sets may be chosen such that each $\varphi_i$ is bi-Lipschitz  with constant less than $1+\varepsilon$, each $K_i$  is compact  and the sets $\varphi_1(K_1), \varphi_2(K_2),\dots$ are pairwise disjoint.  For such a choice of the $\varphi_i$ and $K_i$ there is an $N=N(\varepsilon)$ such that $$\Hd\left(A\backslash \bigcup_{i=1}^N\varphi_i(K_i)\right) < \varepsilon.$$ 

The following definition of strong $(\Hd, d)$-rectifiability strengthens this condition in that for each $\varepsilon>0$ there must be a finite collection of the mappings as above such that the portion of $A$ not covered by the  union is of strictly lower dimension.  

\begin{definition}\label{df:strongrect}
We say that a set $A\subset\Rp$ is {\em strongly $(\Hd, d)$-rectifiable} if,
for every $\varepsilon>0$, there is a finite collection of compact
subsets of $\Rd$ $K_1,\dots,K_N$ and a corresponding set of
bi-Lipschitz maps $\varphi_1:K_1\to\Rp,\dots,\varphi_N:K_N\to\Rp$
such that
\begin{itemize}
\item[1.] The bi-Lipschitz constant of each map is less than $1+\varepsilon$.
\item[2.] $\Hd(\varphi_i(K_i)\cap\varphi_j(K_j)) = 0$ for all $i\ne j$.
\item[3.] $\dim \left(A\backslash \bigcup_{i=1}^N \varphi_i(K_i)\right) < d$.
\end{itemize}
\end{definition}
Note that compact subsets of $d$-dimensional $C^1$ manifolds are strongly $(\Hd, d)$-rectifiable and any strongly $(\Hd,d)$-rectifiable set is $(\Hd,d)$-rectifiable.

The {\em  Riesz $s$-potential of a measure $\mu\in\MA$ at a point $x\in\Rp$} is given by
$$
U_s^{\mu}(x) := \int \Rk{x}{y}{s}d\mu(y).
$$
If $\mu\in\MAp$, then $\lim_{s\uparrow d} U_s^{\mu}(x)= \infty$ for all $x$ in some set of positive $\mu$-measure 
and hence, as mentioned previously, $\lim_{s\uparrow d} I_s(\mu)=\infty$ (cf. \cite{Mattila1}).
For $x\in\Rp$, we define  (when it exists) the following normalized $d$-potential 
\begin{equation}\label{df:Ud}
\Ud{\mu}(x):=\lim_{s\uparrow d}(d-s)U_s^{\mu}(x).
\end{equation}
In certain cases, $\Ud{\mu}(x)$ behaves like an average density of $\mu$ at $x$.   In particular,  it is shown by Hinz in  \cite{Hinz1} that, if there is a constant $C=C(x)$ such that $\mu(B(x,r)) < Cr^d$ for all $r>0$, then $\Ud{\mu}(x)$ equals $d$ times  the order-two density defined by Bedford and Fisher in~\cite{BedfordFisher1}  
$$
\lim_{\varepsilon\to0}\frac{1}{|\ln \varepsilon|}\int_\varepsilon^1\frac{\mu(B(x,r))}{r^d}\frac{1}{r}\, dr
$$ at any point $x$ where this limit exists. This in turn equals the usual density 
$$D_\mu(x):=\lim_{r\to 0}\frac{\mu(B(x,r))}{r^d}$$ at any point $x$ where $D_\mu(x)$ exists. (Here $B(x,r)$ denotes the closed unit ball centered at $x$ of radius $r$. The corresponding open ball is denoted $B(x,r)^0$.) 
Note that the order-two density exists for many measures for which the density does not  (cf. ~\cite{Zahle3} and references therein). 

 Finally, we remark that  in  \cite{Putinar1} M. Putinar  considered a different normalized Riesz $d$-potential   in his work on solving inverse moment problems.    
\subsection{Main Results}
Our first theorem asserts that the normalized $d$-energy $\Id$ is well defined and gives rise to a minimization problem with a unique solution. Note that we   choose a normalization of $\Hd$ so that the $\Hd$-measure of $B(0,1)\subset\Rd$ is $2^d$. 

\begin{theorem}\label{th:Existence}
Let $A\subset\Rp$ be compact and strongly $(\Hd, d)$-rectifiable such that $\Hd(A)>0$. Let $\lambda^d := \Hdr/\Hd(A)$. Then 
\begin{itemize}
\item[1.]$\Id(\mu)$ exists as an extended real number for every measure $\mu\in\MA$ and 
$$
\Id(\mu)=\left\{ \begin{array}{cc}
2^dd\int\left(\frac{d\mu}{d\Hdr}\right)^{2}d\Hdr & \qquad\text{if }\mu\ll\Hdr,\\ 
\infty & \qquad\text{otherwise.}\end{array}\right.
$$
\item[2.] If, for some measure $\mu\in\MA$ $\Id(\mu)<\infty$, then $\Ud{\mu}$ exists and is finite $\mu$-a.e. and 
$$\Id(\mu) = \int \Ud{\mu} d\mu.$$
\item[3.] $\Id(\lambda^d) < \Id(\nu)$ for every measure $\nu \in \MAp\backslash\{\lambda^d\}$.
\end{itemize}
\end{theorem}

The second theorem asserts the weak-star convergence of the $s$-equilibrium measures to normalized Hausdorff measure as $s$ approaches $d$ from below. The essential idea behind the proof is that any weak-star limit point of the $s$-equilibrium measures, as $s$ approaches $d$, has normalized $d$-energy less then or equal to that of $\lambda^d$.

\begin{theorem}\label{th:Convergence}
Let $A\subset\Rp$ be compact and strongly $(\Hd, d)$-rectifiable such that $\Hd(A)>0$. Let $\lambda^d := \Hdr/\Hd(A)$. Then $\mu^s
\cws\lambda^d$ as $s\uparrow d$.
\end{theorem}

The remainder of this paper is organized as follows. In Section~\ref{sec:Existence} we prove several lemmas leading to a proof of Theorem~\ref{th:Existence}. In Section~\ref{sec:Convergence} we show that $\mu^s$ converges to $\lambda^d$ first, for the simpler case where $A$ is a $d$-dimensional compact subset of $\Rd$. Then, by gluing together near isometries of compact subsets of $\Rd$, the theorem is proven for the more general case where $A$ is a strongly $(\Hd, d)$-rectifiable subset of $\Rp$.

\section{The Existence of a Unique Minimizer of $\Id$}\label{sec:Existence}

In this paper, the Fourier transform of a finite Borel measure $\mu$ supported on $\Rd$ is defined by
$$ \Rd \ni \xi \to \hat{\mu}(\xi):=\int_{\Rd} e^{-2\pi ix\cdot\xi}d\mu(x).
$$  For a compactly supported Radon measure $\mu$ on $\Rd$ and $s\in(0,d)$ the
Riesz $s$-energy of $\mu$ may be expressed as
(cf. \cite{Landkof1,Mattila1,Wolff1})
$$
I_{s}(\mu)=c(s,d)\int_{\Rd}|\xi|^{s-d}|\hat{\mu}(\xi)|^{2}d\Ld(\xi),$$
 where  $ \Ld $ denotes  Lebesgue measure  on
$\Rd$ and the constant $c(s,d)$ is given by \[
c(s,d) = \pi^{s-\frac{d}{2}}\frac{\Gamma(\frac{d-s}{2})}{\Gamma(\frac{s}{2})}.\]
Observe that (cf. \cite[ch. 1]{Landkof1}) \begin{equation}
\lim_{s\uparrow d}\ (d-s)c(s,d)=\omega_{d},\label{eq:DminusStimesCSD}\end{equation}
 where $\omega_{d}$ is the surface area of the $d-1$ sphere in $\Rd$.

\begin{lemma}\label{FlatEnergy}Let $K\subset\Rd$ be compact. 
For a measure $\mu\in{\mathcal M}(K)$ we have
$$
\Id(\mu)=\omega_{d}\|\hat{\mu}\|_{2,\Ld}^{2}.
$$ Further, if $\Id(\mu)<\infty$, then $\mu\ll\Ld$.

\end{lemma}

\begin{proof}For any measure $\mu\in{\mathcal M}(K)$ the Riesz $s$-energy can be expressed as\[
\Is(\mu)=c(s,d)\int_{|\xi|\leq1}|\xi|^{s-d}|\hat{\mu}(\xi)|^{2}d\Ld(\xi)+c(s,d)\int_{|\xi|>1}|\xi|^{s-d}|\hat{\mu}(\xi)|^{2}d\Ld(\xi).\]
By dominated convergence\[ \lim_{s\uparrow d}\
\int_{|\xi|\leq1}|\xi|^{s-d}|\hat{\mu}(\xi)|^{2}d\Ld(\xi)=\int_{|\xi|\leq1}|\hat{\mu}(\xi)|^{2}d\Ld(\xi),\]
and by monotone convergence\[ \lim_{s\uparrow d}\
\int_{|\xi|>1}|\xi|^{s-d}|\hat{\mu}(\xi)|^{2}d\Ld(\xi)=\int_{|\xi|>1}|\hat{\mu}(\xi)|^{2}d\Ld(\xi).\]
From \eqref{eq:DminusStimesCSD} the first statement is proven.

An established result (cf. \cite{Wolff1}) is that, if $\hat{\mu}\in L^{2}(\Ld)$,
then $\mu\ll\Ld$ and $d\mu/d\Ld\in L^{2}(\Ld)$.
\end{proof}

\begin{definition}(cf. \cite[ch. 1]{Mattila1}) Let $\mu$ be a
compactly supported Radon measure on $\Rp$ and let
$\varphi~:~\supp\{\mu\}~\rightarrow~\Rpp$ be continuous. The
\emph{image measure} associated with $\mu$ and $\varphi$ is the
set-valued function $\varphi_{\#}\mu$ defined by
$$
\varphi_{\#}\mu(E):=\mu(\varphi^{-1}(E)).
$$
 \end{definition} The following are straightforward consequences
of the above definition.

\begin{itemize}
\item [1.] $\varphi_{\#}\mu$, as defined above, is a compactly supported
Radon measure on $\Rpp$. 
\item [2.] For a non-negative $\varphi_{\#}\mu$-measurable function
$f$ 
$$
\int fd\varphi_{\#}\mu=\int f(\varphi)d\mu.
$$
\end{itemize}
For $A\subset\Rp$,  a bi-Lipschitz map $\varphi:A\rightarrow \Rpp$ with constant $L$, and 
a measure  $\mu\in\MA$  it
  follows   that 
\begin{equation}
\label{eq:ImageMeasureEnergy}
\frac{1}{L^{s}}\Is(\varphi_{\#}\mu)\leq\Is(\mu)\leq L^{s}\Is(\varphi_{\#}\mu),
\end{equation}
and 
\begin{equation}
\label{eq:LipschitzHausdorffBound}
\frac{1}{L^d}\Hd(\varphi(A)) \leq \Hd(A) \leq L^d \Hd(\varphi(A)).
\end{equation}
Note \eqref{eq:LipschitzHausdorffBound} implies $\mu\perp\Hd$ if  and only if  $\varphi_{\#}\mu\perp\Hd$.

\begin{lemma}
\label{lemma:AbsContFiniteDEn}
Let $A\subset \Rp$ be strongly $(\Hd, d)$-rectifiable and let $\mu\in\MA$ be such that $\mu\nll\Hdr$, then $\Id(\mu)$ exists and is infinite.
\end{lemma}

\begin{proof}
Let $\mu\in\MA$ such that $\mu\nll\Hdr$. Let $\mu = \mu^{\perp} + \mu^{\ll}$ be the Lebesgue decomposition of $\mu$ with 
respect to $\Hdr$. Let $K_1,\dots,K_N$ and $\varphi_1:K_1\to\Rp,\dots,\varphi_N:K_N\to\Rp$ be the compact subsets of $\Rd$ 
and the corresponding maps with bi-Lipschitz constant less than $2$ provided by the strong $(\Hd, d)$-rectifiability of $A$. Let $B = 
A \backslash \bigcup_{i=1}^N\varphi_i(K_i)$ and $s_0=\dim B$. If $\mu(B)>0$, then, by the equality of the capacitory and 
Hausdorff dimensions (cf.~\cite{Mattila1}), $I_s(\mu)=\infty$ for all $s\in(s_0,d)$. Hence $\Id(\mu)=\infty$. 

If $\mu(B)=0$, then 
$$
0<\mu^{\perp}(A)\leq\sum_{i=1}^N\mu^{\perp}(\varphi_i(K_i)).
$$
Choose $j\in 1,\ldots, N$ such that $\mu^{\perp}(\varphi_j(K_j))>0$, and define $\nu_j := \mu^{\perp}_{\varphi_j(K_j)}$. Since $
\nu_j \perp \Hd_{\varphi_j(K_j)}$, it follows that $\varphi_{j\#}^{-1} \nu_j \perp \Hd$ and hence $\varphi_{j\#}^{-1} \nu_j \perp \Ld$. 
By Lemma~\ref{FlatEnergy} we have that $\Id(\varphi_{j\#}^{-1} \nu_j)=\infty$ and by \eqref{eq:ImageMeasureEnergy} it follows 
that $\infty=\Id(\varphi_{j\#}\varphi_{j\#}^{-1} \nu_j) = \Id(\nu_j) \leq \Id(\mu)$.
\end{proof}

\begin{lemma}
\label{lemma:UdRN}
Let $A\subset \Rp$ be strongly $(\Hd, d)$-rectifiable and $\mu\in\MA$ such that $\mu\ll\Hdr$. Then 
$$
\Ud{\mu} = 2^dd\frac{d\mu}{d\Hdr} \qquad \text{$\Hdr$-a.e.}
$$
\end{lemma}

\begin{proof}
As already noted, because $A$ is strongly $(\Hd, d)$-rectifiable, it is $(\Hd,d)$-rectifiable.  For any $(\Hd,d)$-rectifiable set, a density result (cf~\cite[ch. 16]{Mattila1}) and the Radon-Nikod\'ym Theorem give
$$
\lim_{r\downarrow 0} \frac{\Hdr(B(x,r))}{(2r)^d} = 1 \qquad \text{and} \qquad \lim_{r\downarrow 0} \frac{\mu(B(x,r))}{\Hdr(B(x,r))} =\left.\frac{d\mu}{d\Hdr}\right|_x<\infty\qquad \text{for $\Hdr$-a.e. $x$.}
$$ 
For $\Hdr$-a.e. $x$ we then have $\sup_{r>0} \mu(B(x,r))/r^d < \infty $ and
$$
D_\mu(x) = 2^d\left.\frac{d\mu}{d\Hdr}\right|_x.
$$
Hence the order-two density exists and, by the result of Hinz in~\cite{Hinz1} mentioned earlier, $$\Ud{\mu}(x) =  2^dd\left.\frac{d\mu}{d\Hdr}\right|_x \qquad \text{$\Hdr$-a.e.}$$

\end{proof}

\begin{lemma}
\label{lemma:L2Argument}
Let $A\subset \Rp$ be strongly $(\Hd, d)$-rectifiable and let $\mu\in\MA$ be such that $\mu\ll\Hdr$ and $d\mu/d\Hdr \notin L^2(\Hdr)$, then $\Id(\mu)$ exists and is infinite.
\end{lemma}

\begin{proof}
From Lemma~\ref{lemma:UdRN} and Fatou's Lemma we immediately obtain
\begin{eqnarray*}
\infty &=& 2^dd \int \left(\frac{d\mu}{d\Hdr}\right)^2 d\Hdr = \int \left(2^dd\frac{d\mu}{d\Hdr} \right)d\mu = \int\Ud{\mu}d\mu\\
&=&\int\left(\lim_{s\uparrow d} (d-s)\int\Rk{x}{y}{s}d\mu(y)\right)d\mu(x) \leq \liminf_{s\uparrow d} (d-s)\iint\Rk{x}{y}{s}d\mu(y)d\mu(x).
\end{eqnarray*}
\end{proof}

\begin{lemma}
\label{lemma:UpperAhlfors}
Let $A\subset \Rp$ be strongly $(\Hd, d)$-rectifiable. There is a constant $C$ depending only on $A$ such that for all $x\in\Rp$ and all 
$r>0$
$$
\frac{\Hdr(B(x,r))}{r^d} < C.
$$
\end{lemma}
\begin{proof}
Let $K_1,\dots,K_N$ and $\varphi_1:K_1\to\Rp,\dots,\varphi_N:K_N\to\Rp$ be the compact subsets of $\Rd$ and the 
corresponding maps with bi-Lipschitz constant less than $2$ provided by the strong $(\Hd, d)$-rectifiability of $A$. Since $\Hd(A) = 
\Hd(\bigcup_{i=1}^N\varphi_i(K_i))$ and since each $\varphi_i$ is bijective, we have
$$
\frac{\Hdr(B(x,r))}{r^d} \leq \sum_{i=1}^N \frac{\Hd(\varphi_i(K_i)\cap B(x,r))}{r^d}=
\sum_{i=1}^N \frac{\Hd(\varphi_i(K_i\cap \varphi_i^{-1}(B(x,r))))}{r^d} \leq 
\sum_{i=1}^N \frac{2^d\Hd(K_i\cap \varphi_i^{-1}(B(x,r)))}{r^d},
$$
where the last inequality follows from \eqref{eq:LipschitzHausdorffBound}. Since $\Hd(K_i\cap \varphi_i^{-1}(B(x,r)))\leq2^{2d}r^d
$, the claim holds with $C=2^{3d}N$.
\end{proof}

\begin{lemma}
\label{lemma:MaxArgument}
Let $A\subset \Rp$ be strongly $(\Hd, d)$-rectifiable and $\mu\in\MA$ be such that $\mu\ll\Hdr$ and $d\mu/d\Hdr \in L^2(\Hdr)$, then $\Id(\mu)$ exists and
$$\Id(\mu)=\int\Ud{\mu}d\mu.$$
\end{lemma}
\begin{proof}
The maximal function of $\mu$ with respect to $\Hdr$ may be expressed 
as
$$
M_{\Hdr}\mu(x) := \sup_{r>0}\frac{\mu(B(x,r))}{\Hdr(B(x,r))} = \sup_{r>0}\frac{1}{\Hdr(B(x,r))} \int_{B(x,r)} \frac{d\mu}{d\Hdr} d\Hdr.
$$
The maximal function maps $L^2(\Hdr)$ to itself and so $M_{\Hdr}\mu(x)\in L^2(\Hdr)$. 

We construct a $\mu$-integrable function that bounds $(d-s)U_s^{\mu}$ for all $s\in(0,d)$. Lemma~\ref{lemma:UdRN} holds $
\mu$-a.e. and, for an $x$ for which Lemma~\ref{lemma:UdRN} holds, we follow an argument found in~\cite[ch. 2]{Mattila1} to obtain 
\begin{eqnarray}
\label{eq:MaxBound1}
\nonumber
(d-s)\int\Rk{x}{y}{s}d\mu(y) &=&  (d-s)\int_0^\infty \mu\left(\left\{y\in\Rp : \Rk{x}{y}{s} > t\right\}\right) dt\\
\nonumber
&=&  (d-s)s\int_0^\infty \frac{\mu(B(x,r))}{r^{s+1}} dr\\
&=& (d-s)s\int_0^{\diam A} \frac{\mu(B(x,r))}{\Hdr(B(x,r))}\frac{\Hdr(B(x,r))}{r^d}r^{d-s-1}dr\\
\label{eq:MaxBound2}
&+&(d-s)s\int_{\diam A}^\infty\frac{\mu(B(x,r))}{r^{s+1}}dr.
\end{eqnarray}
The right hand side of \eqref{eq:MaxBound1} is bounded by $CM_{\Hdr}\mu(x) s(\diam A)^{d-s}$, where $C$ is the constant 
established in Lemma~\ref{lemma:UpperAhlfors}. The quantity in \eqref{eq:MaxBound2} is bounded by $(d-s)\mu(\Rp)(\diam A)^{-
s}$. We may maximize these bounds over $s\in[0,d]$ to obtain a bound $(d-s)U_s^{\mu}$ of the form $C_1 M_{\Hdr}
\mu(x)  + C_2\mu(\Rp)$. The $\mu$-integrability of this bound is established via the Cauchy-Schwarz inequality as follows
$$
\int\left( C_1 M_{\Hdr}\mu(x)  + C_2\mu(\Rp) \right)d\mu  \leq C_1\left\|M_{\Hdr}\mu(x) \right\|_{2,\Hdr}
\left\|\frac{d\mu}{d\Hdr} \right\|_{2,\Hdr} +
C_2\mu(\Rp)^2 < \infty.
$$
By dominated convergence the claim follows.
\end{proof}

\subsection{Proof of Theorem 1.2}

\begin{proof}[Proof of Theorem~\ref{th:Existence}]
Let $A$ satisfy the hypotheses of Theorem~\ref{th:Existence}. The first two claims of the theorem are proven in lemmas \ref{lemma:AbsContFiniteDEn}, \ref{lemma:L2Argument} and \ref{lemma:MaxArgument}. 

Let $\nu$ denote the finite measure $(2^dd)^{-1}\Hdr$. The set of measures with finite normalized $d$-energy are identified with 
the non-negative cone in $L^2(\nu)$ (denoted $L^2(\nu)_+$) via the map $\mu\leftrightarrow d\mu/d\nu$. Under this map we have $\Id(\mu)=\left\| d\mu/d
\nu\right\|_{2,\nu}^2$. A measure $\mu$ of finite $d$-energy is a probability measure if  and only if  $\left\|d\mu/d\nu\right\|_{1,\nu} = 1$. The last claim in the theorem is proven by finding a unique, non-negative function $f$ that minimizes $\|\cdot\|_{2,\nu}$ subject to the constraint $\|f\|_{1,\nu}=1$. We address this problem using the following, standard Hilbert space argument. 

The non-negative constant function $1/\nu(\Rp)$ satisfies the constraint $\|1/\nu(\Rp)\|_{1,\nu}=1$. Let $f\in L^2(\nu)_+$ be such that $\|f\|_{1,\nu}=1$ and $\|f\|_{2,\nu} \le \|1/\nu(\Rp) \|_{2,\nu}$, then
$$
\frac{1}{\nu(\Rp)}=\left\|\frac{f}{\nu(\Rp)}\right\|_{1,\nu} = \left\langle f, \frac{1}{\nu(\Rp)} \right\rangle_\nu 
\leq \|f\|_{2,\nu}\left\|\frac{1}{\nu(\Rp)}\right\|_{2,\nu} 
\leq \left\|\frac{1}{\nu(\Rp)}\right\|^2_{2,\nu} = \frac{1}{\nu(\Rp)}.
$$
Thus
$$
\left\langle f, \frac{1}{\nu(\Rp)} \right\rangle_\nu = \|f\|_{2,\nu}\left\|\frac{1}{\nu(\Rp)}\right\|_{2,\nu}.
$$
From the Cauchy-Schwarz inequality $f=1/\nu(\Rp)$ $\nu$-a.e. By the identification above, the measure, $\lambda^d := \Hdr/\Hd(A)\in\MAp$, uniquely minimizes $\Id$ over $\MAp$.
\end{proof}

\section{The Weak-Star Convergence of $\mu^s$ to $\lambda^d$}\label{sec:Convergence}

\begin{lemma} \label{OnePlusEtaSimple} Let $K\subset\Rd$ be a compact
set. Then, for every $\eta>0$, there is an $s_{0}=s_{0}(\eta)$
such that, for any $s$ and $t$ satisfying $s_{0}<s<t<d$ and any measure
$\mu\in\mathcal{M}(K)$, \[
(d-s)\Is(\mu)\leq(1+\eta)\left[(d-t)I_{t}(\mu)+\eta\mu(\Rd)^{2}\right].\]
 \end{lemma} 
\begin{proof} 
If $I_s(\mu)=\infty$, then $I_t(\mu)=\infty$ for $t>s$ and the lemma
holds trivially.  Now suppose that $I_s(\mu)<\infty$ for some $s$ such that $(d-t)c(t,d) > \omega_d/2$ for all $t\in(s,d)$ and observe that
\begin{eqnarray}
\nonumber
(d-s)I_{s}(\mu) & = & (d-s)c(s,d)\int_{\Rd}|\xi|^{s-d}|\hat{\mu}(\xi)|^{2}d\Ld(\xi)\\
\label{eq:e6} 
& = & \frac{(d-s)c(s,d)}{(d-t)c(t,d)}(d-t)c(t,d)\int_{\Rd}|\xi|^{s-d}|\hat{\mu}(\xi)|^{2}d\Ld(\xi).
 \end{eqnarray} 
We may approximate the integral in
 \eqref{eq:e6} as follows. \begin{eqnarray*} & &
 \int_{\Rd}|\xi|^{s-d}|\hat{\mu}(\xi)|^{2}d\Ld(\xi)\\ & = &
 \int_{|\xi|\leq1}|\xi|^{s-d}|\hat{\mu}(\xi)|^{2}d\Ld(\xi)+\int_{|\xi|>1}|\xi|^{s-d}|\hat{\mu}(\xi)|^{2}d\Ld(\xi)\\
 & \leq &
 \int_{|\xi|\leq1}(|\xi|^{s-d}-|\xi|^{t-d})|\hat{\mu}(\xi)|^{2}d\Ld(\xi)+\int_{|\xi|\leq1}|\xi|^{t-d}|\hat{\mu}(\xi)|^{2}d\Ld(\xi) + \int_{|\xi|>1}|\xi|^{t-d}|\hat{\mu}(\xi)|^{2}d\Ld(\xi)\\ 
 & \leq &
 \mu(\Rd)^{2}\int_{|\xi|\leq1}(|\xi|^{s-d}-|\xi|^{t-d})d\Ld(\xi)+\int_{\Rd}|\xi|^{t-d}|\hat{\mu}(\xi)|^{2}d\Ld(\xi).\end{eqnarray*}
By \eqref{eq:DminusStimesCSD} we may pick $s_0\in(0,d)$ high enough so that, for any $s$ and $t$ satisfying $s_0<s<t<d$ $$
 \frac{(d-s)c(s,d)}{(d-t)c(t,d)}<1+\eta, \qquad (d-t)c(t,d) < 2\omega_d,$$ and $$
 \left|\int_{|\xi|\leq1}(|\xi|^{s-d}-|\xi|^{t-d})d\Ld(\xi)\right|<\frac{\eta}{2\omega_d}.$$
 \end{proof}

The following generalization of Lemma \ref{OnePlusEtaSimple} will be applied
repeatedly to measures supported on the bi-Lipschitz image of a compact set, $K\subset\Rd$. Let $\mu\in{\mathcal M}(\varphi(K))$ be such a measure. Using
\eqref{eq:ImageMeasureEnergy} to bound the $s$-energy of $\varphi_{\#}^{-1}\mu$,
applying Lemma \ref{OnePlusEtaSimple} to $\varphi_{\#}^{-1}\mu$, and then
using \eqref{eq:ImageMeasureEnergy} again to bound the $t$-energy of the measure $\varphi_{\#}\varphi_{\#}^{-1}\mu=\mu$ we 
obtain the following.

\begin{corollary} \label{OnePlusEtaSingleEmbed} Let $K\subset\Rd$ be a compact
set and suppose $\varphi:K\to \Rp$ is bi-Lipschitz with constant $L$. Then, for every $\eta>0$ there is an $s_{0}=s_{0}(\eta)$
such that for any $s$ and $t$ satisfying $s_{0}<s<t<d$ and any
measure  
$\mu\in\mathcal{M}(\varphi(K))$, we have 
$$
(d-s)\Is(\mu)\leq L^{d}(1+\eta)\left[L^{d}(d-t)I_{t}(\mu)+\eta\mu(\Rp)^{2}\right].$$
 \end{corollary}

In the proof of the following proposition we shall use the Principle of Descent (cf. \cite[ch.1 \S 4]{Landkof1}), a consequence of which is that, if $s
\in(0,d)$ and if a sequence of compactly supported Radon measures $\{\mu_n\}_{n=1}^\infty$ converges in the weak-star 
topology to $\psi$, then $I_s(\psi) \le \liminf_{n\to\infty}\ I_s(\mu_n)$. 

Proposition~\ref{FlatConvergence} is a simple case of Theorem \ref{th:Convergence} and its proof illustrates the approach used in the proof of Theorem~\ref{th:Convergence}.

\begin{proposition}\label{FlatConvergence}

Let $A\subset\Rd$ be a compact set such that $\Hd(A)>0$. Let $\mu^{s}$
denote the $s$-equilibrium measure supported on $A$. Then $\mu^{s}\cws\lambda^d:=\Hdr/\Hd(A)$ as $s\uparrow d$.

\end{proposition}

\begin{proof}

Let $\psi\in\MAp$ be a weak-star cluster point of $\mu^{s}$ as $s\uparrow d$.
Let $\{s_{n}\}_{n=1}^{\infty}\uparrow d$ such that $\mu^{s_{n}}\cws\psi$
as $n\rightarrow\infty$. Let $\eta>0$ be arbitrary, $s_{0}$ be
as provided by Lemma \ref{OnePlusEtaSimple}, and let $s\in(s_{0},d)$.
We have \begin{eqnarray*}
(d-s)\Is(\psi) & \leq & \liminf_{n\rightarrow\infty}\ (d-s)\Is(\mu^{s_{n}})\\
 & \leq & \liminf_{n\rightarrow\infty}\ (1+\eta)\left[(d-s_{n})I_{s_{n}}(\mu^{s_{n}})+\eta\right]\\
 & \leq & \liminf_{n\rightarrow\infty}\ (1+\eta)\left[(d-s_{n})I_{s_{n}}(\lambda^d)+\eta\right]\\
 & = & (1+\eta)\left[\Id(\lambda^d)+\eta\right],\end{eqnarray*}
 where the first inequality is an application of the Principle of
Descent. The second
inequality follows from Lemma \ref{OnePlusEtaSimple} where $t$ in
the statement of the lemma is chosen to be $s_{n}$, and the third
from the minimality of $I_{s_{n}}(\mu^{s_{n}})$. 

The variable $s$ may be taken arbitrarily close to $d$, and so $\Id(\psi)\leq(1+\eta)[\Id(\lambda^d)+\eta]$.
The variable $\eta$ was also chosen arbitrarily and we conclude $\Id(\psi)\leq\Id(\lambda^d)$.
Theorem \eqref{th:Existence} ensures that $\lambda^d$ is the unique probability measure that
minimizes $\Id$, and so $\psi=\lambda^d$. Since this holds
for any weak-star cluster point, the proposition is proven.
\end{proof}

The rest of the paper shall employ several classical results
from potential theory (cf. \cite{Landkof1}). Let $\Es$ denote the set
of all signed Radon measures supported in $\Rp$ of finite total variation such that
$\mu$ is an element of $\Es$ if  and only if  $\Is(|\mu|)<\infty$. The set
$\Es$ is a vector space, and, when combined with the following bilinear form
$$\Is(\mu,\nu)=\iint\Rk{x}{y}{s}d\mu(x)d\nu(y),$$
is a pre-Hilbert space. Further, for $\mu\in\Es$
$$\Is(\mu)=\int U_{s}^{\mu}d\mu.$$ A property is said to hold
\emph{approximately everywhere}, if it holds everywhere except on a
set of points contained in a compact set that supports no non-trivial
measures in $\Es$. For $s<\dim A$, the
equilibrium measure $\mu^s$ satisfies $U_{s}^{\mu^{s}}=I_{s}(\mu^{s})$
approximately everywhere in $\supp\left\{ \mu^{s}\right\} $. In
particular $U_{s}^{\mu^{s}}=I_{s}(\mu^{s})$ $\mu^{s}$-a.e.

The proof of Theorem \ref{th:Convergence} follows essentially the same
approach used in the proof of Proposition \ref{FlatConvergence}. The
only technical hurdle is to establish an analog of Lemma
\ref{OnePlusEtaSimple} for the case when $A$ is strongly $(\Hd, d)$-rectifiable and of lower dimension
than that of the embedding space, $\Rp$. This is accomplished by breaking $A$ into near isometries of compact subsets of $\Rd$, establishing the desired estimate one each piece, and showing that the pieces can be glued back together without affecting the estimate. This is the content of lemmas \ref{SmallSubsetsHaveSmallEnergy}, \ref{DisjointPiecesPlusSmallEnergy} and \ref{OnePlusEtaEmbedded}.

\begin{lemma}\label{SmallSubsetsHaveSmallEnergy}

Let $A\subset\Rp$ be a compact, strongly $(\Hd, d)$-rectifiable set such that
$\Hd(A)>0$.  Let $K\subset\Rd$ be compact, and
$\varphi:K\rightarrow\Rp$ a bi-Lipschitz map such that
$\varphi(K)\subset A$. Then, for every $\varepsilon>0$, there is
an $s_{0}=s_{0}(\varepsilon)$ and a constant
$C_{K,\varphi}=C_{K,\varphi}(A,K,\varphi)$ such that, for any Borel set
$B\subset\Rp$ satisfying $\Hdr(\partial B)=0$ and any
$s\in(s_{0},d),$
$$ \limsup_{t\uparrow d}\
(d-s)I_{s}\left(\mu_{B\cap\varphi(K)}^{t}\right)\leq
C_{K,\varphi}\sqrt{\Hdr(B)}+\varepsilon.
$$
The boundary, $\partial B$, is computed in the usual topology on $\Rp$.

\end{lemma}

\begin{proof}
Without loss of generality assume $\varepsilon \in (0,1)$. Let $B\subset \Rp$ be a Borel set such that $\Hdr(\partial B)=0$. 
Observe that
\begin{equation}\label{eq:energyproductbound}
I_{t}\left(\mu_{B\cap\varphi(K)}^{t}\right)=\int_{B\cap\varphi(K)}U_{t}^{\mu_{B\cap\varphi(K)}^{t}}d\mu^{t}\leq\int_{B\cap\varphi(K)}
U_{t}^{\mu^{t}}d\mu^{t}=I_{t}(\mu^{t})\mu^{t}(B\cap\varphi(K)).
\end{equation}
We bound the quantity $\limsup_{t\uparrow d}\ \mu^{t}(B\cap\varphi(K))$ as
follows. Let $\psi\in\MA$ be a weak-star cluster point of
$\mu_{B\cap\varphi(K)}^{t}$ as $t\uparrow d$, and let $\left\{
t_{n}\right\} _{n=1}^{\infty}\uparrow d$ such that
$\mu_{B\cap\varphi(K)}^{t_{n}}\cws\psi$ as $n\rightarrow\infty$. Let
$L$ denote the bi-Lipschitz constant of $\varphi$. Choose
$\tilde{s}_{0}$ so that Corollary~\ref{OnePlusEtaSingleEmbed} applied
to Radon measures with supported in $\varphi(K)$ holds for $\eta=1$. Let
$\lambda^d:=\Hdr/\Hd(A)$ denote the minimizer of $\Id$ over
$\MAp$.  For any $s\in(\tilde{s}_{0},d)$, 
\begin{eqnarray*}
(d-s)I_{s}(\psi) & \leq & \liminf_{n\rightarrow\infty}\
(d-s)I_{s}\left(\mu_{B\cap\varphi(K)}^{t_{n}}\right)\\ & \leq &
\liminf_{n\rightarrow\infty}\
2L^{d}\left[(d-t_{n})L^{d}I_{t_{n}}(\mu^{t_{n}})+1\right]\\ & \leq &
\liminf_{n\rightarrow\infty}\
2L^{d}\left[(d-t_{n})L^{d}I_{t_{n}}(\lambda^d)+1\right]\\ & = &
2L^{2d}\Id(\lambda^d)+2L^{d}=:M<\infty.
\end{eqnarray*} 
The first inequality follows from the Principle of Descent, the second from
Corollary~\ref{OnePlusEtaSingleEmbed} and the inequality, $I_s(\mu^{t_n}_{B\cap\varphi(K)}) \le I_s(\mu^{t_n})$, and the third from the
minimality of $I_{t_{n}}(\mu^{t_{n}})$. Letting $s\uparrow d$ we see
that, for any weak-star cluster point $\psi$ of
$\mu_{B\cap\varphi(K)}^t$ (as $t\uparrow d$), $\Id(\psi)\leq M$.  Theorem \ref{th:Existence}
ensures that $\psi\ll\Hdr$, and so $\psi(\partial B)=0$, implying 
$\mu^{t_n}(B\cap\varphi(K))=\mu_{B\cap\varphi(K)}^{t_{n}}(B)\rightarrow\psi(B)$ as
$n\rightarrow\infty$. 

The set $\overline{B}\cap A$ is strongly
$d$-rectifiable, and if $\psi(B)>0$, then $\Hdr(B)>0$, implying
$\Hd\left(\overline{B}\cap A\right)>0$ and by Theorem \ref{th:Existence},
$\Id$ is minimized over $\mathcal{M}_{1}\left(\overline{B}\cap
A\right)$ by $\lambda^{d,\overline{B}\cap A} :=
\Hdrest{\overline{B}\cap A}/\Hd\left(\overline{B}\cap
A\right)$. We then have 
$$ 
\frac{2^{d}d}{\Hdr(B)}
=\frac{2^{d}d}{\Hdr\left(\overline{B}\right)}
=\frac{2^{d}d}{\Hd\left(\overline{B}\cap A\right)}
=\Id\left(\lambda^{d,\overline{B}\cap A}\right)
\leq\Id\left(\frac{\psi}{\psi\left(\overline{B}\right)}\right) = \Id\left(\frac{\psi}{\psi(B)}\right)
\leq\frac{M}{\psi(B)^{2}},
$$
and we may conclude 
$$\psi(B)\leq\sqrt{\frac{M}{2^{d}d}\Hdr(B)}.$$
(If $\psi(B)=0$, then the above inequality holds trivially.)  It follows from the above inequality and \eqref{eq:energyproductbound}
that for any Borel set $B\subset\Rp$ with $\Hdr(\partial B)=0$ we have
\begin{equation}\label{initialresult}
\limsup_{t\uparrow d}\ (d-t)I_{t}\left(\mu_{B\cap\varphi(K)}^{t}\right)
\leq
\limsup_{t\uparrow d}\ (d-t)I_{t}(\mu^{t})\limsup_{t\uparrow d}\ \mu^{t}(B\cap \varphi(K))
\leq
\Id(\lambda^d)\sqrt{\frac{M}{2^{d}d}}\sqrt{\Hdr(B)}.
\end{equation}

We complete the proof of this lemma by appealing to
Corollary~\ref{OnePlusEtaSingleEmbed} applied to measures supported on $\varphi(K)$ with
$\eta=\varepsilon/2L^{d}$. If $s_{0}$ is chosen so that
Corollary~\ref{OnePlusEtaSingleEmbed} holds, then, for any
$s\in(s_{0},d)$ and $t\in(s,d)$, \begin{eqnarray*}
(d-s)I_{s}\left(\mu_{B\cap\varphi(K)}^{t}\right) & \leq &
L^{d}\left[\left(1+\frac{\varepsilon}{2L^{d}}\right)\left(L^{d}(d-t)I_{t}\left(\mu_{B\cap\varphi(K)}^{t}\right)+\frac{\varepsilon}{2L^{d}}
\right)\right]\\
& \leq &
2L^{2d}(d-t)I_{t}\left(\mu_{B\cap\varphi(K)}^{t}\right)+\varepsilon.\end{eqnarray*}
Taking the limit superior of both sides as $t\uparrow d$ and appealing
to \eqref{initialresult} completes the proof with
$C_{K,\varphi}=2L^{2d}\Id(\lambda^d)\sqrt{M/2^{d}d}$.
\end{proof}

\begin{lemma}

\label{DisjointPiecesPlusSmallEnergy} Let $A\subset\Rp$ be a compact,
strongly $(\Hd, d)$-rectifiable set such that $\Hd(A)>0$. Then, for every
$\varepsilon>0$, there exists a finite collection of compact
subsets of $\Rd$ $\tilde{K}_{1},\dots,\tilde{K}_{N}$ and a
corresponding set of bi-Lipschitz maps
$\tilde{\varphi}_{1}:\tilde{K}_{1}\rightarrow\Rp,\dots,\tilde{\varphi}_{N}:\tilde{K}_{N}\rightarrow\Rp$
each with bi-Lipschitz constant less than $1+\varepsilon$, such that
\begin{itemize}
\item[1.]
 $\tilde{\varphi}_{i}(\tilde{K}_{i})\cap\tilde{\varphi}_{j}(\tilde{K}_{j})=\emptyset$ for  $i\neq j$, and 
 \item[2.] there is an $s_{0}=s_{0}
(\varepsilon)\in(0,d)$,
such that for $
\tilde{B}:=A\backslash\bigcup_{i=1}^{N}\tilde{\varphi}_{i}(\tilde{K}_{i})$ and all $s\in(s_{0},d)$ we have  \[
\limsup_{t\uparrow d}\ (d-s)I_{s}(\mu_{\tilde B}^{t})\leq\frac{\varepsilon}{N}.\]
\end{itemize}
 \end{lemma}

\begin{proof}

Without loss of generality assume $\varepsilon\in(0,1)$. Since $A$ is strongly $(\Hd, d)$-rectifiable, we
may find a set, $A_{0}\subset\Rp$, compact sets $K_{1},\dots,K_{N}
\subset\Rd$ and bi-Lipschitz maps $\varphi_{1}:K_1\to\Rp,\dots,\varphi_{N}:K_N\to\Rp$ with constant less than $1+\varepsilon$ such
that $A=\bigcup_{i=1}^N\varphi_i(K_i) \cup A_0$, where $\dim A_0<d$, and $\Hd(\varphi_i(K_i)\cap\varphi_j(K_j))=0$. Let
$\delta=\varepsilon^2/4N^2\in(0,1)$.  The
set $E=\bigcup_{i\neq
j}^{}\left(\varphi_{i}(K_{i})\cap\varphi_{j}(K_{j})\right)$ is a
compact set of $\Hdr$-measure $0$. Since $\Hdr$ is Radon, there is an
open set $\mathcal{O}$ such that $E\subset\mathcal{O}$ and
$\Hdr(\mathcal{O})<\delta N^{-4}\left(\max\left\{
C_{K_{1},\varphi_{1}},\dots,C_{K_{N},\varphi_{N}}\right\}
\right)^{-2}$ where $C_{K_{i},\varphi_{i}}$ is the constant provided
by Lemma~\ref{SmallSubsetsHaveSmallEnergy} applied to
$\varphi_{i}(K_{i})\subset A$.

For any point $x\in E$, we may find a non-empty open ball
$B(x,R)^{0}\subset\mathcal{O}$.  Since $\partial
B(x,r_{1})\cap\partial B(x,r_{2})=\emptyset$ for any $r_1 \ne r_2$
and since $\Hdr$ is a finite measure, all but a countable set of
values of $r\in(0,R)$ must be such that $\Hdr(\partial
B(x,r))=0$. Construct an open cover of $E$ as follows.\[
\Omega=\left\{ B(x,r)^{0}:x\in E,\
B(x,r)^{0}\subset\mathcal{O},\:\Hdr\left(\partial
B(x,r)\right)=0\right\} .\] Choose a finite sub-cover
$\Omega'\subset\Omega$, of $E$. Let $B=\bigcup_{b\in\Omega'}b$.
Since $\partial B \subset \bigcup_{b\in\Omega'} \partial b$, we have
that $\Hdr(\partial B)=0$. Let $B_{i}=B\cap\varphi_{i}(K_{i})$.  For
any $s$, $t\in(0,d)$ with $t>\max\left\{
s,\dim A_{0}\right\} $ we have, by the equality of the Hausdorff and capacitory dimensions, that $\mu^t(A_0)=0$ and
hence
$$
(d-s)I_{s}(\mu_{B}^{t})
\leq
(d-s)I_{s}\left(\mu^t_{A_0} + \sum_{i=1}^{N}\mu_{B_{i}}^{t}\right)=\sum_{i,j=1}^{N}(d-s)I_{s}(\mu_{B_{i}}^{t},\mu_{B_{j}}^{t}).
$$ By Jensen's inequality followed by the Cauchy-Schwarz inequality
applied to the inner-product $I_{s}(\cdot,\cdot)$ we have \[
\left[\frac{1}{N^{2}}\sum_{i,j=1}^{N}(d-s)I_{s}(\mu_{B_{i}}^{t},\mu_{B_{j}}^{t})\right]^{2}\leq\frac{1}{N^{2}}\sum_{i,j=1}^{N}\left[(d-s)
I_{s}(\mu_{B_{i}}^{t},\mu_{B_{j}}^{t})\right]^{2}\leq\frac{1}{N^{2}}\sum_{i,j=1}^{N}(d-s)I_{s}(\mu_{B_{i}}^{t})(d-s)I_{s}(\mu_{B_{j}}^
{t}).\]
Let $s_{0}=\max\left\{
\dim A_{0},s_{0,1},\dots,s_{0,N}\right\} $, where
$s_{0,i}$ is the value of $s_{0}$ provided by Lemma
\ref{SmallSubsetsHaveSmallEnergy} applied to
$\varphi_{i}(K_{i})\subset A$, and where the value of $\varepsilon$ in
the statement of Lemma \ref{SmallSubsetsHaveSmallEnergy} is chosen to
be $\delta/N^{2}$.  Combining the previous bounds gives, for
$s\in(s_{0},d)$,\begin{eqnarray*} \left[\limsup_{t\uparrow d}\
(d-s)I_{s}(\mu_{B}^{t})\right]^{2} & \leq &
N^{2}\sum_{i,j=1}^{N}\limsup_{t\uparrow d}\
(d-s)I_{s}(\mu_{B_{i}}^{t})\limsup_{t\uparrow d}\
(d-s)I_{s}(\mu_{B_{j}}^{t})\\ & \leq &
N^{2}\sum_{i,j=1}^{N}\left(C_{K_{i},\varphi_{i}}\sqrt{\frac{\delta}{N^{4}\left(C_{K_{i},\varphi_{i}}\right)^{2}}}+\frac{\delta}{N^{2}}
\right)\left(C_{K_{j},\varphi_{j}}\sqrt{\frac{\delta}{N^{4}\left(C_{K_{j},\varphi_{j}}\right)^{2}}}+\frac{\delta}{N^{2}}\right)\\
& = &
N^{2}\sum_{i,j=1}^{N}\left(\frac{\sqrt{\delta}+\delta}{N^{2}}\right)^{2}\\
& \leq &
4\delta=\left(\frac{\varepsilon}{N}\right)^{2}.\end{eqnarray*}
The value of $s_{0}$, the set $\tilde{B}:=\left(B\cap A\right)\cup A_0$, the compact sets
$\tilde{K}_{i}:=K_{i}\backslash\varphi_{i}^{-1}(B)$, and the
bi-Lipschitz maps $\tilde{\varphi}_i := \varphi_{i}|_{\tilde{K}_i}$
satisfy the properties claimed in the lemma for the value of
$\varepsilon$ given.
\end{proof}
\begin{lemma}\label{OnePlusEtaEmbedded}

Let $A\subset\Rp$ be a strongly $(\Hd, d)$-rectifiable, compact set such
that $\Hd(A)>0$. Then, for every $\eta>0$, there is an $s_{0}=s_{0}(\eta)$,
such that for all $s\in(s_{0},d)$ we have
$$
\limsup_{t\uparrow d}\ (d-s)I_{s}(\mu^{t})\leq(1+\eta)\limsup_{t\uparrow d}\ (d-t)I_{t}(\mu^{t})+\eta.
$$

\end{lemma}

\begin{proof}

Let $\lambda^d:=\Hdr/\Hd(A)$ denote the unique minimizer of $\Id$ over $\MAp$. Let   $\eta>0$. Choose    $\varepsilon\in(0,1)$ such that
\begin{equation}\label{eq:messyCond}
\max\left\{\left(\varepsilon\left[2+(1+\varepsilon)^{d+1}\right]+2\sqrt{\varepsilon(1+\varepsilon)^{2d+1}\Id(\lambda^d)+\varepsilon^
{2}(1+\varepsilon)^{d+1}}\right),\, \left((1+\varepsilon)^{2d+1}-1\right)\right\}<\eta.
\end{equation}

From Lemma \ref{DisjointPiecesPlusSmallEnergy} there is an
$s_{1}\in(0,d)$, a sequence of compact sets
$\tilde{K}_{1},\dots,\tilde{K}_{N}\subset\Rd$ and a sequence of
bi-Lipschtiz maps
$\tilde{\varphi}_{1}:\tilde{K}_{1}\to\Rp,\dots,\tilde{\varphi}_{N}:\tilde{K}_{N}\to\Rp$
each with constant less than $1+\varepsilon$ such that
$\tilde{\varphi}_{i}(\tilde{K}_{i})\cap\tilde{\varphi}_{j}(\tilde{K}_{j})=\emptyset$
for $i\neq j$, and
$\tilde{B}:=A\backslash\bigcup_{i=1}^{N}\tilde{\varphi}_{i}(\tilde{K}_{i})$
satisfies the following for all $s\in(s_{1},d)$
$$\limsup_{t\uparrow d}\
(d-s)I_{s}(\mu_{\tilde B}^{t})\leq\frac{\varepsilon}{N}.
$$

For $s\in(s_{1},d)$ we
have\begin{eqnarray}
\nonumber
\limsup_{t\uparrow d}\ (d-s)I_{s}(\mu^{t}) & = & \limsup_{t\uparrow d}\ (d-s)I_{s}\left(\mu_{\tilde{B}}^{t}+\sum_{i=1}^{N}\mu_{\tilde
{\varphi}_{i}(\tilde{K}_{i})}^{t}\right)\\
 & \le & \limsup_{t\uparrow d}\ (d-s) I_{s}\left(\mu_{\tilde{B}}^{t}\right)\label{eq:T0}\\
 & + & 2\limsup_{t\uparrow d}\ \sum_{i=1}^{N}(d-s)I_{s}\left(\mu_{\tilde{B}}^{t},\mu_{\tilde{\varphi}_{i}(\tilde{K}_{i})}^{t}\right)\label
{eq:T1}\\
 & + & \limsup_{t\uparrow d}\ \sum_{\stackrel{i,j=1}{i\neq j}}^{N}(d-s)I_{s}\left(\mu_{\tilde{\varphi}_{i}(\tilde{K}_{i})}^{t},\mu_{\tilde
{\varphi}_{j}(\tilde{K}_{j})}^{t}\right)\label{eq:T2}\\
 & + & \limsup_{t\uparrow d}\ \sum_{i=1}^{N}(d-s)I_{s}\left(\mu_{\tilde{\varphi}_{i}(\tilde{K}_{i})}^{t}\right).\label{eq:T3}\end
{eqnarray}

We next find upper bounds for each of the terms in (\ref{eq:T0}--\ref{eq:T3}). 
First,  Lemma \ref{DisjointPiecesPlusSmallEnergy} implies that, for  $s\in(s_{1},d)$, expression  \eqref{eq:T0} is less than  $\varepsilon/N$ .

Second, using  Jensen's inequality and the Cauchy-Schwarz inequality
in the same manner as in the proof of Lemma \ref{DisjointPiecesPlusSmallEnergy}
we have
$$
\sum_{i=1}^{N}(d-s)I_{s}\left(\mu_{\tilde{B}}^{t},\mu_{\tilde{\varphi}_{i}(\tilde{K}_{i})}^{t}\right)\leq\sqrt{N(d-s)I_{s}
\left(\mu_{\tilde{B}}^{t}\right)\sum_{i=1}^{N}(d-s)I_{s}\left(\mu_{\tilde{\varphi}_{i}(\tilde{K}_{i})}^{t}\right)}.
$$
Since each $\tilde{\varphi}_i$ is bi-Lipschitz with constant
$(1+\varepsilon)$, Corollary~\ref{OnePlusEtaSingleEmbed} (with the
values of $\eta$ and $L$ as stated in the corollary chosen to be $\varepsilon$
and $1+\varepsilon$ respectively) ensures that there is some
$s_2\in (s_1,d)$ such that, for $s_2<s<t<d$, we have
\begin{equation}
(d-s)I_{s}\left(\mu_{\tilde{\varphi}_{i}(\tilde{K}_{i})}^{t}\right)\leq(1+\varepsilon)^{2d+1}(d-t)I_{t}\left(\mu_{\tilde{\varphi}_{i}(\tilde{K}_
{i})}^{t}\right)+\varepsilon(1+\varepsilon)^{d+1}\mu_{\tilde{\varphi}_{i}(\tilde{K}_{i})}^{t}(\Rp)^{2}.\label{eq:e7}
\end{equation}
Then \eqref{eq:e7}, together with the bound for \eqref{eq:T0},  implies that expression \eqref{eq:T1} is bounded above 
by
\begin{equation}\label{eq:T1bnd1}
 2\sqrt{N\frac{\varepsilon}{N}\limsup_{t\uparrow d}\ 
 \left[
(1+\varepsilon)^{2d+1} \sum_{i=1}^N(d-t)I_t\left(\mu_{\tilde{\varphi}_{i}(\tilde{K}_{i})}^{t}\right) + 
\varepsilon(1+\varepsilon)^{d+1}\sum_{i=1}^N\mu_{\tilde{\varphi}_{i}(\tilde{K}_{i})}^{t}(\Rp)
 \right]}
%
%
 \end{equation}
Using 
$$ \limsup_{t\uparrow d}\ \sum_{i=1}^N (d-t)I_{t}\left(\mu_{\tilde{\varphi}_{i}(\tilde{K}_{i})}^{t}\right)\le  \limsup_{t\uparrow d}\ 
(d-t)I_{t}(\mu^{t})\le  \limsup_{t\uparrow d}\ (d-t)I_t(\lambda^d)=\Id(\lambda^d)$$ 
it follows that, for $s\in (s_2,d)$, expression  \eqref
{eq:T1} is bounded above by
  $$2\sqrt{\varepsilon\left[(1+\varepsilon)^{2d+1}\Id(\lambda^d)+\varepsilon(1+\varepsilon)^{d+1}\right]}.$$

We bound \eqref{eq:T2} as follows. For $ 1 \le i\neq j\le N$, let
$D_{i,j}=\dist(\tilde{\varphi}_{i}(\tilde{K}_{i}),\tilde{\varphi}_{j}(\tilde{K}_{j}))>0$
and let $s_{i,j}\in (0,d)$ be such that
$(d-s)D_{i,j}^{-s}\leq\varepsilon/N^{2}$ for all
$s\in(s_{i,j},d)$. For such an $s$,
$(d-s)I_{s}(\nu_{1},\nu_{2})\leq \nu_{1}(\Rp)\nu_{2}(\Rp)\varepsilon/N^{2}$,
for any $\nu_{1}$, $\nu_{2}\in \MA$ supported on
$\tilde{\varphi}_{i}(\tilde{K}_{i})$ and
$\tilde{\varphi}_{j}(\tilde{K}_{j})$ respectively. Let
$s_{0}:=\max\left\{ s_{2},s_{i,j}:i\neq j\right\} $.  For all
$s\in(s_0,d)$, 
$$
\sum_{\stackrel{i,j=1}{i\neq j}}^{N}(d-s)I_{s}\left(\mu_{\tilde{\varphi}_{i}(\tilde{K}_{i})}^{t},\mu_{\tilde{\varphi}_{j}
(\tilde{K}_{j})}^{t}\right)<\varepsilon.
$$

From \eqref{eq:e7} we have the following bound for \eqref{eq:T3}
\begin{eqnarray*}
\sum_{i=1}^{N}(d-s)I_{s}\left(\mu_{\tilde{\varphi}_{i}(\tilde{K}_{i})}^{t}\right)
& \leq &
(1+\varepsilon)^{2d+1}\left(\sum_{i=1}^{N}(d-t)I_{t}\left(\mu_{\tilde{\varphi}_{i}(\tilde{K}_{i})}^{t}\right)\right)+\varepsilon(1+
\varepsilon)^{d+1}\left(\sum_{i=1}^{N}\mu_{\tilde{\varphi}_{i}(\tilde{K}_{i})}^{t}(\Rp)^{2}\right)\\
& \leq &
(1+\varepsilon)^{2d+1}(d-t)I_{t}(\mu^{t})+\varepsilon(1+\varepsilon)^{d+1}.\end{eqnarray*}
For $s\in(s_{0},d)$, the preceding estimates, together with
\eqref{eq:messyCond}, gives 
\begin{eqnarray*}
\limsup_{t\uparrow d}\ (d-s)I_{s}(\mu^{t}) &\leq & \left[\varepsilon\left[2+(1+\varepsilon)^{d+1}\right]+2\sqrt{\varepsilon(1+\varepsilon)^{2d+1}\Id(\lambda^d)+ \varepsilon^{2}(1+ \varepsilon)^{d+1}}\right]\\
&+&\left[(1+\varepsilon)^{2d+1}\right]\limsup_{t\uparrow d}\ (d-t)I_{t}(\mu^{t}) \\
&\le & \eta+(1+\eta)\limsup_{t\uparrow d}\ (d-t)I_{t}(\mu^{t}).
\end{eqnarray*}
\end{proof}

\subsection{Proof of Theorem 1.3}

\begin{proof}[proof of theorem \ref{th:Convergence}]

Let $A$ satisfy the hypotheses of Theorem \ref{th:Convergence} and hence of Theorem \ref{th:Existence}. Let $\lambda^d:=
\Hdr/\Hd(A)$ denote the unique minimizer of $\Id$ over $\MAp$. Let
$\psi$ be any weak-star cluster point of $\mu^{s}$ as $s\uparrow d$,
and let $\left\{ s_{n}\right\} _{n=1}^{\infty}\uparrow d$ such that
$\mu^{s_{n}}\cws\psi$. Let $\eta>0$ be arbitrary. Let $s_{0}$ be
the value provided by lemma \ref{OnePlusEtaEmbedded} for this choice of $\eta$. For any $s\in(s_{0},d)$,
we have\begin{eqnarray*}
(d-s)I_{s}(\psi) & \leq & \liminf_{n\rightarrow\infty}\ (d-s)I_{s}(\mu^{s_{n}})\\
 & \leq & \limsup_{n\rightarrow\infty}\ (d-s_{n})I_{s_{n}}(\mu^{s_{n}})(1+\eta)+\eta\\
 & \leq & \limsup_{n\rightarrow\infty}\ (d-s_{n})I_{s_{n}}(\lambda^d)(1+\eta)+\eta\\
 & = & (1+\eta)\Id(\lambda^d)+\eta.\end{eqnarray*}
 As in the proof of Proposition \ref{FlatConvergence}, the first
inequality follows from the Principle of Descent, the second from
Lemma \ref{OnePlusEtaEmbedded}, and the third from the minimality
of $I_{s_{n}}(\mu^{s_{n}})$. Since $s$ may be chosen arbitrarily close
to $d$, $\Id(\psi)\leq(1+\eta)\Id(\lambda^d)+\eta$. Since $\eta$
was also arbitrarily chosen, $\Id(\psi)\leq\Id(\lambda^d)$. The uniqueness of the minimizer $\lambda^d$
ensured by Theorem \ref{th:Existence} proves that $\psi=\lambda^d$ and is sufficient to prove Theorem~\ref{th:Convergence}.
\end{proof}

\bibliography{References}
\bibliographystyle{abbrv} 

\end{document}